\documentclass[twoside,11pt]{article}

\usepackage{jmlr2e}

\usepackage{graphicx}
\usepackage{algorithmic}
\usepackage{amsmath}
\usepackage{amssymb}
\usepackage{amsfonts}

\usepackage{color}
\usepackage[colorlinks=true,
            linkcolor=red,
            urlcolor=blue,
            citecolor=blue]{hyperref}
\usepackage{natbib}
\usepackage{algorithm} 
\usepackage[all]{xy}
\usepackage{color}
\usepackage{DEF-bold}
\usepackage{DEF-script}

\def\pa{{\rm pa}}

\def\ch{{\rm ch}}

\def\tmax{t_{\rm max}}

\def\V{{V}}

\def\Pr{\mathbb{P}}

\def\e{{\rm e}}

\def\p{^\prime}

\def\d{{\rm d}}

\ShortHeadings{Particle Gibbs for MJPs}{Miasojedow and Niemiro}
\firstpageno{1}

\begin{document}

\title{Particle Gibbs algorithms for Markov jump processes}

\author{\name B{\l}a{\.z}ej Miasojedow \email bmia@mimuw.edu.pl	 \\
       \name Wojciech Niemiro \email wniem@mimuw.edu.pl\\
        \addr Institute of Applied Mathematics, University of Warsaw\\ 
Banacha 2, 02-097 Warsaw,  Poland\\ }

\editor{}

\maketitle

\begin{abstract}
In the present paper we propose a new MCMC algorithm for sampling from the posterior distribution of hidden trajectory of 
a Markov jump process. Our algorithm is based on the idea of exploiting virtual jumps, introduced  by \citet{RaoTeh2013a}. 
The main novelty is that our algorithm uses particle Gibbs with ancestor sampling (PGAS, see \citet{andrieu2010particle,lindsten2014particle}) to update the skeleton, 
while Rao and Teh use forward filtering backward sampling (FFBS). In contrast to previous methods our algorithm can be 
implemented even if the state space is infinite.  In addition, the cost of a single step of the  proposed algorithm 
does not depend on the size of the state space. The computational cost of our methood is of order 
$\mathcal{O}(N\mathbb{E}(n))$, where $N$ is the  number of particles used in the PGAS algorithm and 
$\mathbb{E}(n)$ is the expected number of jumps (together with virtual ones).
The cost of the  algorithm of Rao and Teh is of order $\mathcal{O}(|\mathcal{X}|^2\mathbb{E}(n))$, where $|\mathcal{X}|$
is the size of the state space. Simulation results show that our algorithm with PGAS converges slightly slower than the algorithm with FFBS, 
if the size of the state space is not big. However, if the size of the state space increases, the proposed method outperforms existing ones. 
We give special attention to a hierarchical version of our algorithm which can be applied to continuous time Bayesian networks (CTBNs).
\end{abstract}


\section{Introduction}
Markov jump processes (MJP) are natural extension of Markov chains to continuous time. 
They are widely applied in modelling of the phenomena of chemical, biological,
economic and other sciences.  An important class of MJP are continuous time Bayesian networks (CTBN) 
introduced by \citet{Sch} under the name of composable Markov chains and then reinvented
by \citet{Nod1} under the current name. Roughly, a CTBN is a multivariate MJP in which the dependence structure 
between coordinates can be described by a graph. Such a graphical representation 
allows for decomposing a large intensity matrix into smaller conditional intensity matrices.   

In many applications it is necessary to consider a situation where the trajectory of a Markov jump process 
is not observed directly, only partial and noisy observations are available. Typically, 
the posterior distribution over trajectories is then analytically intractable. 
The present paper is devoted to MCMC methods for sampling from the posterior in such a situation. 

In the literature there exist several approaches 
to the above mentioned problem: based on sampling 
\citep{BoysWilkKirk2008,EFK,FaSh,FarSher2006,hobolth2009,Nod2,RaoTeh2013a,rao2012mcmc,CTBNMet2014}, based on 
numerical approximations \citep{cohn2010mean,Nod1,Nod3,opper2008variational}. Some of these methods are 
inefficient, like modification of likelihood weighting \citep{Nod2}. 
Other approaches involve expensive computations like matrix exponentiation, spectral decomposition of matrices, 
finding roots of equations. There are also approximate algorithms based on time 
discretization. To the best of our knowledge the most general, efficient and exact method is that proposed by 
\citet{RaoTeh2013a}, and extended to a more general class
of continuous time discrete systems in \citet{rao2012mcmc}. Their algorithm is based on introducing so-called 
virtual jumps and a thinning procedure for Poisson processes. In our approach we combine this method 
with particle MCMC discovered by \citet{andrieu2010particle}. More precisely, instead of forward filtering 
backward sampling algorithm used in the original version, we use particle Gibbs \citep{andrieu2010particle}
with added ancestor resampling  proposed in  \citet{lindsten2012ancestor,lindsten2014particle}.
The proposed method is computationally less expensive. 
Moreover, our algorithm can be directly applied when the state space is infinite, 
in opposition to \citet{rao2012mcmc,RaoTeh2013a}. 


 
\section{Markov jump processes}\label{sec:MPJ}

Consider a continuous time stochastic process $\{X(t),{t\geq 0}\}$  defined on a probability space 
$(\Omega,\mathcal{F},\mathbb{P})$ with a discrete state space $\mathcal{X}$.
Assume the process is time-homogeneous Markov with transition probabilities
\[P^t(s,s^\prime)=\mathbb{P}(X(t+u)=s^\prime|X(u)=s)\;,\]
for $s,s^\prime\in\mathcal{X}$. The initial distribution is denoted by $\nu(s)=\mathbb{P}(X(0)=s)$. 
Since $\mathcal{X}$ is discrete,  $\nu$ can be viewed as a vector and $P^t$ as a matrix
(both possibly infinite). The intensity matrix is defined as follows
\[ Q(s,s^\prime)=\lim_{t\to0}\frac{1}{t}[P^t(s,s^\prime)-I(s,s^\prime)]\;,\]
where $I=P^0$ is the identity matrix. Equivalently, $Q(s,s^\prime)$ is the intensity of jumps from $s$ to $s^\prime$, i.e.
\begin{align*}
 \mathbb{P}\left(X(t+ \d t)=s^\prime| X(t)=s\right)&=Q(s,s^\prime)\d t \text{ for } s\neq s^\prime\;,\\
 \mathbb{P}\left(X(t+\d t)=s| X(t)=s\right)&=1-Q(s)\d t\;, 
 \end{align*}
where $Q(s)=-Q(s,s)=\sum_{s^\prime\neq s} Q(s,s^\prime)$ denotes the intensity of leaving state $s$. Clearly we have 
$\sum_{s^\prime} Q(s,s^\prime)=0$ for all $s\in\mathcal{X}$. 
Throughout this paper we assume that $Q$ is non-explosive \citep{Nor98}, which means that almost surely only finite number of jumps occur 
in any bounded time interval $[0,\tmax]$.
This assumption is fulfilled in most applications we have in mind. Sufficient and necessary conditions for $Q$ to be non-explosive can be found in 
\citep{Nor98}[Thm. 2.7.1, 2.7.2, Cor. 2.7.3]. 
From now on the interval  $[0,\tmax]$ is fixed. Let there be $m$ jumps and let these jumps occur at ordered moments $T=(t_1,\dots, t_m)$.
Moments of jumps $T$ with a corresponding sequence of states, denoted by $S=(s_0,s_1,\ldots,s_m)=(X(0),X(t_1),\dots,X(t_m))$, 
fully describe the trajectory $X([0,\tmax])$.
By definition of the  process $\{X(t)\}$, every interval between jumps, $t_j-t_{j-1}$ for $j=1,\dots,m$, with the convention $t_0=0$, is distributed according 
to the exponential distribution with parameter $Q(s_{j-1})$. The skeleton $S$ is a discrete time Markov chain with initial distribution $\nu$
and transition matrix given by
\[
\begin{cases}
 \displaystyle\frac{Q(s,s^\prime)}{Q(s)}&\text{if } s\neq s^\prime\;;\\
 0&\text{if } s= s^\prime\;.
\end{cases}
\]
Thus random variable $(T,S)$ has density
\begin{align}\label{eq:density}
 p(T,S)&=\nu(s_0)\prod_{j=1}^m Q(s_{j-1})\exp\left\{-Q(s_j)(t_j-t_{j-1})\right\}\frac{Q(s_{j-1},s_j)}{Q(s_{j-1})} 
                   \exp\left\{-Q(s_m)(\tmax-t_m)\right\}\nonumber\\
 &=\nu(s_0)\prod_{j=1}^m Q(s_{j-1},s_j)\exp\left\{-Q(s_j)(t_j-t_{j-1})\right\}\exp\left\{-Q(s_m)(\tmax-t_m)\right\}\;,
\end{align}
where $m=|T|$. The last factor $\exp\left\{-Q(s_m)(\tmax-t_m)\right\}$ comes from the fact that the waiting time for jump 
$m+1$ can be arbitrary but greater than $\tmax-t_m$.  

\section{Virtual jumps}\label{sec:Virtual}

Let $\{X(t)\}$ be a homogeneous Markov process with intensity matrix $Q$  and let $R(s)\geq Q(s)$ for all $s\in\mathcal{X}$. 
Consider the following sampling scheme \citep{rao2012mcmc}, based on dependent thinning, i.e.\ rejection sampling for 
an inhomogeneous Poisson process \citep{1979}.
We generate a sequence of potential times of jumps $(\tilde{t}_1,\tilde{t}_2,\ldots)$.
For a given moment $\tilde{t}_{k-1}$ and a current value of the process $\tilde{s}_{k-1}=X(\tilde{t}_{k-1})$, 
we draw the next time interval $\tilde{t}_k-\tilde{t}_{k-1}$ from the exponential distribution with parameter $R(\tilde{s}_{k-1})$. 
With probability ${Q(\tilde{s}_{k-1})}/{R(\tilde{s}_{k-1})}$ the process jumps at time $\tilde{t}_k$ to another state, and this new state is 
$\tilde{s}_k$ with probability ${Q(\tilde{s}_{k-1},\tilde{s}_k)}/{Q(\tilde{s}_{k-1})}$. With probability 
$(1-Q(\tilde{s}_{k-1}))/{R(\tilde{s}_{k-1})}$ the process does not jump and
we put $\tilde{s}_{k}=\tilde{s}_{k-1}$. The resulting redundant skeleton $\tilde{S}=(\tilde{s}_0,\tilde{s}_1,\tilde{s}_2,\ldots)$ is therefore a Markov chain with  
transition matrix $P$ defined by
\begin{equation}
 \label{eq:matP}
 P(s,s^\prime)=\begin{cases}
                \displaystyle\frac{Q(s,s^\prime)}{R(s)}&\text{  if } s\neq s^\prime\;;\\
                \\
                1- \displaystyle\frac{Q(s)}{R(s)}&\text{ if } s= s^\prime\;.
               \end{cases}
\end{equation}
We summarize this procedure as the following algorithm~\ref{alg:rej}.

\begin{algorithm}[H]
\caption{Thinning procedure.}
\label{alg:rej}
\begin{algorithmic}
 \STATE Set $\tilde{t}_0=0$ and $k=0$.
 \STATE Draw $\tilde{s}_0\sim\nu(\cdot)$.
 \WHILE{$\tilde{t}_k<\tmax$}
 \STATE Set $k=k+1$.
 \STATE Draw  $W\sim {Exp}(R(\tilde{s}_{k-1}))$.
 \STATE Set $\tilde{t}_k=\tilde{t}_{k-1}+W$.
 \STATE Draw $\tilde{s}_k\sim P(\tilde{s}_{k-1},\cdot\;)$.
 \ENDWHILE
\end{algorithmic}
\end{algorithm}
As before, consider the process in a fixed interval of time $[0,\tmax]$. Let $\tilde{T}=(\tilde{t}_1,\tilde{t}_2,\ldots,\tilde{t}_n)$ be the set of
moments generated by the above algorithm. Let  $J=\{k>1: \tilde s_k\neq \tilde s_{k-1}\}$. Denote by $T=\tilde T_J$ and $V=\tilde T_{-J}$  
the moments of true jumps and virtual jumps, respectively. The process $\{X(t)\}$ resulting from the algorithm has the 
same probability distribution as that in Section \ref{sec:MPJ}. This fact is explicitly formulated 
in Proposition \ref{prop:density} below. 
\begin{proposition}\label{prop:density}
 The marginal distribution of  $(T=\tilde{T}_J,\tilde{S}_J)$ has the density given by \eqref{eq:density}.
\end{proposition}
This proposition is known and can be found e.g.\ in  \citet{rao2012mcmc}. However,  to make the paper self-contained we give a proof in the Appendix.
The following corollary  shows that, conditionally on  $(T,\tilde{S}_J)$, i.e.\ on true jumps and the skeleton, the set of virtual jumps $V$ 
is an inhomogeneous Poisson  process with intensity $R(X(t))-Q((X(t))$. 
\begin{corollary}\label{cor:density} Let $V_j$ denote moments of virtual jumps between two adjacent true jumps $t_{j-1}$ and $t_j$. 
 The conditional density of $V_j$ is given by
\[
p(V_j|X(t_{j-1})=s,t_{j-1},t_{j})=(R(s)-Q(s))^{|V_j|}\exp\left\{-(t_{j}-t_{j-1}) (R(s)-Q(s)) \right\}\;.
\]
 \end{corollary}
The proof is also given in the Appendix. 
In the sequel we will work with the redundant representation of the process $\{X(t)\}$ introduced
in this section. For simplicity, let us slightly abuse notation and from now on write $S$ instead of
$\tilde{S}$. Clearly, the density of $(T,V,S)$ is given by 
\begin{align}
 \label{eq:density_with_virtual}
 p(T,V,S)&=\nu(s_0)\prod_{k=1}^n R(s_{k-1})P(s_{k-1},s_k)\exp\left\{-\int_0^{\tmax} R(X(u))\d u \right\}\nonumber\\
 &=\nu(s_0)\prod_{k=1}^n(R(s_{k-1})-Q(s_{k-1}))^{\mathbf{1}(s_k=s_{k-1})}Q(s_{k-1},s_k)^{\mathbf{1}(s_k\neq s_{k-1})}\\
 &\qquad\: \times \prod_{k=1}^n \exp\left\{-R(s_k)(\tilde{t}_k-\tilde{t}_{k-1})\right\}\exp\left\{-R(s_n)(\tmax-\tilde{t}_n)\right\}\;,
 \nonumber
\end{align} 
where $n=|T|+|V|$ is the total number of jumps. 

Now we describe two particular choices of intensity $R$. The first one leads to the so-called uniformization \citep{jensen1953markoff,cinlar2013introduction,hobolth2009,RaoTeh2013a}.
Let $\lambda\geq\max_sQ(s)$ and consider moments of potential jumps distributed according to homogeneous Poisson process with intensity $\lambda$. Precisely,
we consider the thinning procedure with $R(s)\equiv\lambda$. Clearly, the joint distribution of true jumps, virtual jumps and 
skeleton is now given by
\begin{equation}
 \label{eq:density_uniformization}
 p(T,V,S)\propto \lambda^n \nu(s_0)\prod_{k=1}^n P(s_{k-1},s_k)\;,
 \end{equation}
where $P$ is the transition matrix of discrete time Markov chain defined by
\begin{equation}
 \label{eq:matP_uniformization}
 P(s,s^\prime)=\begin{cases}
                \displaystyle\frac{Q(s,s^\prime)}{\lambda}&\text{  if } s\neq s^\prime\;;\\
                \\
                1- \displaystyle\frac{Q(s)}{\lambda}&\text{ if } s= s^\prime\;.
               \end{cases}
\end{equation}
In the case of uniformization, conditionally on the trajectory $X([0,\tmax])$, the virtual jumps  form a piecewise homogeneous Poisson process with 
the intensity constant and equal to $\lambda-Q(X(t_j))$ on every time interval $[t_j,t_{j+1})$ for $j=0,1,\dots,|T|$.

The second natural choice is to make virtual jumps distributed as homogeneous Poisson process with intensity $\theta>0$. 
Let $R(s)=Q(s)+\theta$. Then the thinning procedure leads to the following probability  distribution: 
\begin{align}
 \label{eq:density_const}
 p(T,V,S)&\propto \nu(s_0)\prod_{k=1}^n P(s_{k-1},s_k)(\theta+Q(s_{k-1}))\exp\left\{-(Q(s_{k-1})+\theta)(\tilde{t}_k-\tilde{t}_{k-1})\right\}\nonumber\\
         &\qquad\qquad\qquad\qquad \times \exp\left\{-(Q(s_n)+\theta)(\tmax-\tilde{t}_n)\right\}\\
         &=\nu(s_0)\prod_{k=1}^n Q(s_{k-1},s_k)^{\mathbf{1}(s_k\not=s_{k-1})}
            \exp\left\{-Q(s_{k-1})(\tilde{t}_k-\tilde{t}_{k-1})\right\}\nonumber\\
         &  \qquad\qquad\qquad\qquad \times \exp\left\{-Q(s_n)(\tmax-\tilde{t}_n)\right\} \theta^{|V|}\exp\{-\theta\tmax\}\nonumber\;,
 \end{align}   
where the transition matrix $P$ of discrete time Markov chain which generates    
the skeleton $S$ is given by
\begin{equation}
 \label{eq:matP_const}
 P(s,s^\prime)=\begin{cases}
                \displaystyle\frac{Q(s,s^\prime)}{Q(s)+\theta}&\text{  if } s\neq s^\prime\;;\\
                \\
                \displaystyle\frac{\theta}{Q(s)+\theta}&\text{ if } s= s^\prime\;.
               \end{cases}
\end{equation}

\section{Continuous time Bayesian networks}\label{sec:CTBN}

Let $(\mathcal{V},\mathcal{E})$ denote a directed graph with possible cycles.  We write $w\to u$ instead of $(w,u)\in\E$.
For every node $w\in\mathcal{V}$ consider a corresponding space  $\mathcal{X}_w$ of possible states. 
Assume that each space $\mathcal{X}_w$ discrete.
We consider a continuous time stochastic process on the product
space $\mathcal{X}=\prod_{w\in\mathcal{V}} \mathcal{X}_w$. Thus a state $s\in\mathcal{X}$ is a
configuration $s=(s_w)=(s_w)_{w\in\mathcal{V}}$, where $s_w\in\mathcal{X}_w$. 
If $\mathcal{W}\subseteq\mathcal{V}$ then we write $s_\mathcal{W}=(s_w)_{w\in\mathcal{W}}$ for 
configuration $s$ restricted to nodes in $\mathcal{W}$. We also use notation
$\mathcal{X}_\mathcal{W}=\prod_{w\in\mathcal{W}} \mathcal{X}_w$, so that we can write $s_\mathcal{W}\in\mathcal{X}_\mathcal{W}$. 
The set $\mathcal{V}\setminus\{w\}$ will be denoted simply
by $-w$.  We define the set of parents of node $w$ by 
\[\pa(w)=\{u\in\mathcal{V}\;:\;u\to w\}\;,\]
and we define the set of children of node $w$ by   
\[\ch(w)=\{u\in\mathcal{V}\;:\;w\to u\}\;.\]
Suppose we have a family of functions
$Q_w:\X_{\pa(w)}\times(\X_w\times \X_w)\to[0,\infty)$.
For fixed $c\in \X_{\pa(w)}$, we consider $Q_w(c;\cdot,\cdot\;)$ as a  conditional intensity matrix 
(CIM) at node $w$ (only off-diagonal elements of this matrix
have to be specified, the diagonal ones are irrelevant).
The state of a CTBN at time $t$ is a random element $X(t)$ of the space 
$\X$ of configurations. Let $X_w(t)$ denote its $w$th coordinate. 
The process $\left\{X_w(t)_{w\in\mathcal{V}},t\geq 0\right\}$ 
is assumed to be Markov and its evolution can be described 
informally as follows. Transitions at node $w$ depend on the current 
configuration of the parent nodes. If the state
of some parent changes, then node $w$ switches to other transition 
probabilities.  If  $s_w\not=s_w\p$ then  
\begin{equation*}
         \Pr\left(X_w(t+\d t)=s_w\p|X_{-w}(t)=s_{-w},X_w(t)=s_w\right)=
              Q_w(s_{\pa(w)},s_w,s_w\p)\,\d t. 
\end{equation*}
Formally, CTBN is a MPJ with transition intensities given by  
\begin{equation*}
    Q(s,s\p)=
          \begin{cases}
             Q_w(s_{\pa(w)},s_w,s_w\p) & \text{if $s_{-w}=s_{-w}\p$ and $s_{w}\not=s_{w}\p$ for some $w$;} \\        
              0       &  \text{if $s_{-w}\not=s_{-w}\p$ for all $w$,}
          \end{cases}
\end{equation*}
for $s\not=s\p$ (of course, $Q(s,s)$ must be defined ``by subtraction'' to ensure $\sum_{s\p} Q(s,s\p)=0$).  

For a CTBN, the density  of sample path  $X=X([0,\tmax])$ in a bounded time interval $[0,\tmax]$ decomposes as follows:
\begin{equation}
 \label{eq:densCTBN}
 p(X)=\nu(X(0))\prod_{w\in\mathcal{V}}p(X_w\Vert X_{\pa(w)})\;,
\end{equation}
where $\nu$ is the initial distribution on $\X$ and $p(X_w\Vert X_{\pa(w)})$ is the density of piecewise 
homogeneous Markov jump process with intensity matrix equal to $Q_w(c;\cdot,\cdot\;)$ in every time sub-interval 
such that $X_{\pa(w)}=c$.  Formulas for the density of CTBN appear e.g.\ in \citep[Sec.\ 3.1]{Nod2},
\citep[Eq.\ 2]{FaXuSh}, \citep[Eq.\ 1]{FaSh} and \citep{CTBNMet2014}.
These formulas give a factorization of the main part of the density as in \eqref{eq:densCTBN}, but in the first three
of the cited papers the initial distribution $\nu$ is disregarded. There are some subtle problems related to $\nu$, discussed in \cite{CTBNMet2014}. 
Our notation $p(X_w\Vert X_{\pa(w)})$ is consistent with the notion of ``conditioning by intervention'', 
see \e.g.\ \citep{Lauritzen01causalinference}. Indeed, $p(X_w\Vert X_{\pa(w)})$ is the density of the process $X_w$ at note $w$ 
under the assumption that the sample paths at the parent nodes $X_{\pa(w)}$ are fixed and $X_w(0)$ is given, see e.g.\ 
\cite{CTBNMet2014}, for details. 
Below we explicitly write an expression for $p(X_w\Vert X_{\pa(w)})$ in terms of moments of jumps and the 
skeleton of the process $(X_w,X_{\pa(w)})$, as in \eqref{eq:density}.
Let $T^w=(t_0^w\ldots,t_i^w,\ldots)$ and $T^{\pa(w)}=(t_0^{\pa(w)},\ldots,t_j^{\pa(w)},\ldots)$ denote  moments of jumps at node $w\in\V$ and
at parent nodes, respectively. By convention put $t_0^w=t_0^{\pa(w)}=0$ and $t^w_{|T^w|+1}=t^{\pa(w)}_{|T^{\pa(w)}|+1}=\tmax$. Analogously, 
$S^w$ and $S^{\pa(w)}$ denote the corresponding skeletons. Thus we 
divide the time interval $[0,\tmax]$ into segments $[t^{\pa(w)}_j,t^{\pa(w)}_{j+1})$, $j=0,1,\dots |T^{\pa(w)}|$ such that $X_{\pa(w)}$ is constant 
and $X_w$ is homogeneous in each segment. Next we define sets 
$I_j=\{i>0 :\ t^{\pa(w)}_j<t^w_i<t^{\pa(w)}_{j+1}\}$ with notation $j_{\rm{ beg}},j_{\rm {end}}$ for the first and the last element of $I_j$.  
Analogously to \eqref{eq:density}, we obtain the following formula: 
\begin{align}
  \label{eq:conddens}
& p(X_w\Vert X_{\pa(w)})=p(T^w,S^w\Vert S^{\pa(w)},T^{\pa(w)})= \prod_{j=0}^{|T^{\pa(w)}|}\Bigg[ \prod_{i \in I_j}Q_w(s_j^{\pa(w)};s_{i-1}^w,s_i^w)\nonumber\\
 &\times \prod_{i \in I_j\setminus\{j_{\rm {beg}}\}}\exp\left\{-(t_i^w-t_{i-1}^w)Q_w(s_j^{\pa(w)};s_{i-1}^w)\right\}\\
  &\times \exp\left\{-(t_{j_{\rm {beg}}}^w-t_j^{\pa(w)})Q_w(s_j^{\pa(w)};s_{j_{\rm {beg}}-1}^w)-(t_{j+1}^{\pa(w)}-t_{j_{\rm{end}}}^w)Q_w(s_{j}^{\pa(w)};s_{j_{\rm{end}}}^w)\right\} 
  \Bigg]\;.\nonumber
\end{align} 
Formula \eqref{eq:conddens} is equivalent to \citep[Eq.\ 2]{Nod2} and \citep[Eq.\ 1]{FaSh}, but expressed in terms of $(S,T)$.
\bigskip\goodbreak

\section{Hidden Markov models}

Let $X=\{X(t),{0\leq t\leq \tmax}\}$ be a Markov jump process.
Suppose that process $X$ cannot be directly observed but we can observe some random
quantity $Y$ with probability distribution $L(Y|X([0,\tmax])$. Let us say $Y$ is the evidence
and $L$ is the likelihood.  We assume  that the likelihood depends
on $X$ only through the actual sample path  $X([0,\tmax])$ (does not depend on virtual jumps). 
The problem is to restore the hidden trajectory of $X$ given $Y$.  From the
Bayesian perspective, the goal is to compute/approximate the posterior 
\begin{equation}\nonumber
 p(X([0,\tmax])|Y)\propto p(X([0,\tmax]))L(Y|X([0,\tmax])).
\end{equation}
Function $L$, transition probabilities $Q$ and initial distribution $\nu$ are assumed 
to be known. To get explicit form of posterior distribution we will consider two typical forms of noisy observation. 
In the first model, the trajectory of $X$ is observed independently at deterministic time points $t^*_1,\dots,t^*_l$ with corresponding likelihood functions
$L_1(y_1|X(t^*_1)),\dots,L_l(y_l|X(t^*_l))$. Since $X(t)$ is constant between jumps, for all $i=1,\dots,l$
we have $L_i(y_i|X(t^*_i))=L_i(y_i|s_{i^*})$, where $i^*=\sup\{j: t_j\leq t^*_i\}$. 
If the trajectory of $X$ is represented by moments of true jumps $T$, virtual jumps  $V$ and skeleton $S$
then the posterior distribution is given by
\begin{equation}
\label{eq:post_unif}
 p(T,V,S|Y)\propto p(T,V,S)\prod_{i=1}^l L_i(y_i|s_{i^*})
\end{equation} 
where $p(T,V,S)$ is given by \eqref{eq:density_with_virtual}. In Section \ref{sec:Virtual} we considered two scenarios of
adding virtual jumps: via uniformization and via a homogeneous Poisson process. The corresponding densities are given by \eqref{eq:density_uniformization}
and \eqref{eq:density_const}, respectively. Observe that for both variants of adding virtual jumps, the posterior can be expressed in the following form:
\begin{equation}
 \label{eq:skelton_as_hmm}
 p(S|T,V,Y)=\nu(s_0)g_0(s_0)\prod_{k=1}^n P(s_{k-1},s_k) g_k(s_k),
\end{equation}
where $P$ is a Markov transition matrix,  $g_k$ are some functions which can depend on $T,V,Y$ and $n=|T|+|V|$.  
Hence the skeleton, conditionally on all the jumps and the evidence, can be treated as a hidden Markov model with discrete time.

In the second typical model of observation, the evidence $Y$ is a fully observed continuous time stochastic process depending on $X$.
Expressly, we assume that $Y$, given the trajectory of $X$, is a piecewise homogeneous Markov jump process 
such that the pair $(X,Y)$ is a CTBN with the graph structure $X\to Y$. Thus the likelihood can be expressed by \eqref{eq:conddens}, with $X^w=Y$ and 
$X^{\pa(w)}=X$. In this case we can easily obtain the same conclusion as before: for both variants of adding virtual jumps, the posterior can be expressed in the form
\eqref{eq:skelton_as_hmm}. The skeleton is conditionally a hidden Markov model.

\section{MCMC algorithm}

Let us recall the standing assumption that evidence $Y$ depends only on the trajectory of $X$ but not on virtual jumps. This assumption covers most of
usual scenarios and clearly implies that $p(V|T,S,Y)=p(V|T,S)$. Now we are able to state the main algorithm. Similarly to \citep{rao2012mcmc,RaoTeh2013a},
a single step of the iterative procedure is the following. We take  the trajectory
obtained in the previous step, represented by $(T,S)$ (ignoring the virtual jumps). First we sample a new set of virtual jumps $V$. We can use two 
variants of sampling. In the case of uniformization, $V$ is a piecewise homogeneous Poisson process with intensities
$\lambda- Q(X(t_k))$. Alternatively, $V$ is a homogeneous Poisson process with intensity $\theta$. Next we generate a new skeleton $S^\prime$ 
using Markov kernel with $p(S|T,V,Y)$ as invariant distribution.
In \citep{rao2012mcmc,RaoTeh2013a}, the authors use independent sampling by the forward filtering -- backward sampling  (FFBS) algorithm. 
In our approach we propose to use the particle Gibbs algorithm invented by \citet{andrieu2010particle}, which is described in the next section.
Note that a new skeleton with fixed times of potential jumps leads to a new allocation of true and virtual jumps. So we obtain a new trajectory described by
$(T^\prime,V^\prime,S^\prime)$ such that $T\cup V=T^\prime\cup V^\prime$. Finally we remove virtual jumps to obtain a new state $(T^\prime,S^\prime)$. 
The algorithm is summarized below.
\begin{algorithm}
 \caption{Single step of MCMC algorithm.}
 \begin{algorithmic}
  \STATE Input: Previous state $(T,S)$ and observation $Y$.
  \STATE \textbf{1.} Add virtual jumps $V$.
  \STATE \textbf{2.}  Draw new skeleton $S^\prime$ from Markov kernel $K(S,\cdot\;)$ targeting $p(S|T,V,Y)$. Skeleton $S^\prime$ defines 
  new allocation of virtual and true jumps $T^\prime,V^\prime$ such $T\cup V=T^\prime\cup V^\prime$.
 \STATE  \textbf{3.}  Remove virtual jumps $V^\prime$.
 \RETURN new state  $(T^\prime,S^\prime)$.
  \end{algorithmic}
\end{algorithm}

The next proposition shows that this algorithm is ergodic.
\begin{proposition}\label{prop:ergodic} Assume that $\lambda>\max_s Q(s)$ in the case of uniformization or $\theta>0$ in the case of homogeneous virtual 
jumps. Assume that kernel $K(S,S^\prime)$ leaves distribution $p(S|T,V,Y)$ invariant and $K(S,S^\prime)>0$ for all $S^\prime$  such that 
$p(S^\prime|T,V,Y)>0$. Then MCMC algorithm described above is $\phi$-irreducible, aperiodic with stationary distribution
 $\pi(T,S)=p(T,S|Y)$. Thus for $\pi$-almost all initial positions, the algorithm is ergodic, i.e.\
 \[\Vert M((T,S),\cdot\;)^m-\pi(\cdot)\Vert_{\rm {tv}}\overset{m\to\infty}{\longrightarrow} 0\;,\]
 where $M$ denotes the kernel of our MCMC algorithm.
\end{proposition}
\begin{proof}
 By construction it is clear that $p(T,S|Y)$ is a stationary distribution. Since $K(S,S)>0$, with positive probability it happens that the 
 skeleton does not change and hence virtual and true jumps remain unchanged. Clearly $M((T,S),(T,S))>0$ and so Markov chain $M$ is aperiodic.
 The assumption $\lambda>\max_s Q(s)$ or $\theta>0$ ensures that  the step of adding virtual jumps 
 can reach any configuration of virtual jumps. Together with the assumption $K(S,S^\prime)>0$ it leads to the conclusion that all states $(T,S)$ 
 in the support of $\pi$ are reachable. Hence  $M$ is $\phi$-irreducible. It is now enough to invoke 
 the well-known fact that $\phi$-irreducibility and aperiodicity imply ergodicity in total variation norm, 
 see for example \citep[Theorem 4,][]{roberts2004general}. 
\end{proof}
\begin{remark}
 Condition $K(S,S^\prime)>0$  is clearly satisfied by FFBS algorithm, because it is equivalent to independent sampling from  $p(S|T,V,Y)$. 
 This condition is also satisfied by the particle Gibbs algorithm described in the next section.
\end{remark}

In the case of CTBN we can use its dependence structure by introducing a Gibbs sampler over nodes of the graph $(\mathcal{V},\E)$. 
The same idea was exploited in \citep{EFK,RaoTeh2013a}.
By \eqref{eq:densCTBN}, the full conditional distribution of node $w$ given rest of the graph 
has the density 
\begin{equation}
 \label{eq:fullconditionals}
p(X_w|X_{-w},Y)\propto \nu(X_w(0)|X_{-w}(0)) p(X_w\Vert X_{\pa(w)})\prod_{u\in \ch(w)}p(X_u\Vert X_{\pa(u)})L(Y|X)\;.
\end{equation}
The density $p(X_w\Vert X_{\pa(w)})$ corresponds to a piecewise homogeneous Markov process. The expression $\prod_{u\in \ch(w)}p(X_u\Vert X_{\pa(u)})$
can be treated as a part of likelihood, similarly as $L(Y|X)$. Note that the conditional initial distribution $\nu(X_w(0)|X_{-w}(0))$ may be replaced
in formula \eqref{eq:fullconditionals} by the joint initial distribution $\nu(X(0))$, because these two quanities are proportional as functions of $X_w(0)$.
The step which leaves $p(X_w|X_{-w},Y)$ as invariant measure can be realized by the general algorithm described above.
The Gibbs sampler for CTBN is summarized below.
\begin{algorithm}[H]
 \caption{Gibbs sampler for CTBN.}
 \begin{algorithmic}
  \FOR {$w\in\mathcal{V}$ (in a deterministic or random order)}  
  \STATE Simulate $(T_w,S_w)$ using single step of MCMC algorithm targeting $p(X_w|X_{-w},Y)$, with $X_{-w}$ fixed.
  \ENDFOR
  \end{algorithmic}
\end{algorithm}
Note that if observations of different nodes are independent i.e. $L(Y|X)=\prod_{w\in\mathcal{V}}L_w(Y|X_w)$ then the full 
conditional distributions defined by
\eqref{eq:fullconditionals} reduce to 
\begin{equation*}
p(X_w|X_{-w},Y)\propto \nu(X(0)) p(X_w\Vert X_{\pa(w)})L_w(Y|X_w)\prod_{u\in \ch(w)}p(X_u\Vert X_{\pa(u)})\;.
\end{equation*}
Hence within the Gibbs sampler, the step for node $w$ need not involve evaluation of full likelihood.
An immediate corollary from Proposition~\ref{prop:ergodic} is that the Gibbs sampler for CTBN is also ergodic.
Note that in the step of adding virtual jumps we can choose the intensity parameters $\lambda$ or $\theta$ globally but, more generally, we can 
define different intensities for every node $w$, say $\lambda_w$ or $\theta_w$.
\begin{corollary}
 Assume that for every node $w\in\mathcal{V}$ we have $\lambda_w > \max_{s^w,s^{\pa(w)}}Q_v(s^{\pa(w)};s^w)$ in the case of uniformization or $\theta_w>0$ 
 for homogeneous virtual jumps.  Then the Gibbs sampler for CTBN (with either the particle Gibbs or FFBS used in sampling of a new skeleton) is ergodic.
\end{corollary}

\section{Particle Gibbs}\label{sec:pg}

In this section we will use notation which is standard in the literature on sequential Monte Carlo (SMC). Let $s_{0:k}=(s_0,s_1,\ldots,s_k)$. 
Consider a sequence of unnormalized densities $\gamma_k(s_{0:k})$ on increasing product spaces $\mathcal{S}^{k+1}$
for $k=1,\dots n$. The corresponding normalized probability densities are
\[\pi_k(s_{0:k})=\frac{\gamma_k(s_{0:k})}{Z_k}\;,\]
where $Z_k=\int \gamma_k(s_{0:k})\d s_{0:k}$ is the normalizing constant. A special case is the so-called state-space model where
$\pi_k(s_{0:k})=p(s_{0:k}|y_{0:k})$ and $\gamma_k(s_{0:k})=p(s_{0:k},y_{0:k})$. In the sequel we consider the state-space model because 
the distribution of skeleton $S$ given times of jumps (true and virtual; $(T,V)$) fits in this scheme c.f.\ \eqref{eq:skelton_as_hmm}.

Particle MCMC methods introduced by \citet{andrieu2010particle} provide a general framework for constructing an MCMC kernel targeting $\pi(s_{0:k})$ 
with transition rule based on SMC algorithms. Before we describe particle MCMC in detail, we first have to recall standard SMC methods, e.g.\
\citet{Freitas,doucet2009tutorial,del2006sequential,pitt1999filtering}. SMC
sequentially approximates each of probability distributions $\pi_k$ by an empirical distribution
\begin{equation}\label{eq:approxdens}
\hat\pi^N_k(\d s_{0:k})=\sum_{i=1}^N\frac{w_k^i}{\sum_j w_k^j}\delta_{\xi^i_{0:k}}(\d s_{0:k})\;,
 \end{equation}
where $\{\xi_{0:k}^i,w_k^i\}_{i=1}^N$ is a weighted particle system. The system is propagated as follows. 
Given previous approximation 
$\{\xi_{0:k-1}^i,w_{k-1}^i\}_{i=1}^N$ at time $k-1$, first we draw $\{\tilde\xi_{0:k-1}^i\}_{i=1}^N$ from the multinomial distribution with probabilities 
proportional to weights {$\{w_{k-1}^i\}_{i=1}^N$} (this resampling step is introduced to avoid degeneracy of weights).
Next we generate new particles $\{\xi_k^i\}_{i=1}^N$ according to $r_k(\cdot|\tilde\xi_{0:k-1}^i)$ and set $\xi_{0:k}^i=(\tilde \xi_{0:k-1}^i,\xi_k^i)$. 
Here $r_k$ is an instrumental kernel from  $\mathcal{S}^{k}$ to  $\mathcal{S}$ (identified with a conditional density).
Finally we compute new weights 
\begin{equation}
 \label{eq:weights}
 w_k^i=\frac{\gamma_k(\xi^i_{0:k})}{\gamma_{k-1}(\tilde \xi^i_{0:k-1})r_k(\xi_k^i|\tilde\xi_{0:k-1}^i)}\;.
\end{equation}

The SMC algorithm is summarized below.
\begin{algorithm}[H]
 \caption{SMC algorithm.}
 \begin{algorithmic}
 \STATE Initialization
\FOR {$i=1,\dots,N$}
\STATE Draw $\xi_0^i\sim r_0(\cdot)$.
\STATE Compute weights
\[w^i_0=\frac{\gamma_0(\xi_0^i)}{r_0(\xi_0^i)}\;.\]
\ENDFOR
\STATE Main loop
\FOR {{$k=1,\dots,n$}}
 \STATE \textbf{Resampling step:}\\
 Draw ancestors \[\{a^i_k\}_{i=1}^N\sim Multinomial\left(N,\frac{w_{k-1}^1}{\sum_j w_{k-1}^j},\dots,\frac{w_{k-1}^N}{\sum_j w_{k-1}^j}\right)\]
 and set
 \[\tilde\xi_{0:k-1}^i=\xi_{0:k-1}^{a^i_t}\;.\]
 \STATE \textbf{Propagation step:}
 \FOR {$i=1,\dots,N$}
 \STATE Draw \[\xi_k^i\sim r_k(\cdot|\tilde\xi_{0:k-1}^i)\] 
 and set
 \[\xi_{0:k}^i=(\tilde\xi_{0:k-1}^i,\xi_k^i)
 \;.\]
\STATE Compute weights $w_k^i$ from \eqref{eq:weights}.
 \ENDFOR
 \ENDFOR
 \end{algorithmic}
 \end{algorithm}

The particle Gibbs (PGS) algorithm introduced by \citet{andrieu2010particle} is based on SMC algorithm. Given a previous path $s_{0:n}$ we run an SMC algorithm with
one path fixed, say $\xi_{0:n}^N=s_{0:n}$, and obtain system of particles $\{\xi_{0:n}^i,w_n^i\}_{i=1}^N$. Next we draw a new path $s^\prime_{0:n}$ with
probability 
\[\mathbb{P}(s^\prime_{0:n}=\xi_{0:n}^i)\propto w_n^i\;.\]
This procedure defines the following Markov kernel:
\[
K(s_{0:n },\cdot)=\mathbb{E}(\hat \pi^N_n(\cdot)|s_{k:n}=\xi_{0:n}^N)\;,
\]
where $\hat\pi^N_n$ is defined by  \eqref{eq:approxdens}. It is shown that for every number of particles $N$ larger than one, $\pi_n$ is a stationary distribution for kernel $K$ \citep{andrieu2010particle}.

There is a well-known problem of path degeneracy of SMC samplers \citep{doucet2009tutorial}. 
For large $n$,  the beginning of the path $s_{0:n}^\prime$ can be based on only few trajectories. This may lead 
to  poor mixing of the particle Gibbs sampler, because kernel $K$ with high probability leaves the beginning of 
trajectory unchanged,
 see \citep{chopin2013particle,lindsten2013backward,lindsten2014particle}. A remedy for this problem can be an additional step of ancestor 
resampling proposed by \citet{lindsten2014particle}. Let $\xi^N_{0:n}=s_{0:n}$ be the fixed trajectory in the particle Gibbs algorithm. For every $k=1,\dots,n$ we sample
an ancestor of $\xi^N_k$ from the set  of trajectories $\{\xi_{0:k-1}^i\}_{i=1}^N$ with probabilities proportional to weights {
\begin{equation}
 \label{eq:ancestor_weights}
 w_{k-1|n}^i=w_{k-1}^i\frac{\gamma_n((\xi^i_{0:k-1},s_{k:n}))}{\gamma_k(\xi^i_{0:k-1})}\;.
\end{equation} }
The above modification does not change the invariant measure of the Markov kernel \citep[Theorem 1,][]{lindsten2014particle}. 

Now we are ready to describe precisely the step of sampling a new skeleton in the MCMC algorithm for hidden Markov jump processes.
Let us recall \eqref{eq:skelton_as_hmm}. The conditional distribution of skeleton given moments of jumps and evidence is of the form {
\begin{equation*}
 p(S|T,V,Y)=\nu(s_0)g_0(s_0)\prod_{k=1}^n P(s_{k-1},s_k) g_k(s_k)\;,
\end{equation*}
where $P$ is a transition matrix of some Markov chain.
A standard choice is to use priors as instrumental kernels $r$ i.e.\ $r_0=\nu$ and $r_k=P$ for $k>0$. This choice leads to a simplified form of weights:
\begin{equation*}
 w_0^i=g_0(\xi_0^i)\;,\quad w_k^i=g_k(\xi_k^i)\;,\quad   w_{k-1|n}^i=w_{k-1}^i P(\xi_{k-1}^i,s_k)\;,
\end{equation*}
for $i=1\dots,N$ and $k=1,\dots,n$. If the algorithm is applied in a CTBN setting within a Gibbs sampler step, then the initial
conditional distribution $\nu(X_w(0)|X_{-w}(0))$ in \eqref{eq:fullconditionals} might be difficult to sample from. Then we 
can use a different instrumental distribution $r_0$ and compute weights  $w_0^i=g_0(\xi_0^i)\nu(\xi_0^i)/r_0(\xi_0^i)$.
However, in many  scenarios the initial configuration is deterministic ($\nu$ is concentrated at a single configuration) and then this 
problem disappears.
In algorithm~\ref{alg:pgas} we summarize the particle Gibbs with ancestor sampling (PGAS).
\begin{algorithm}[hp!]
 \caption{PGAS for sampling the new skeleton.}
 \label{alg:pgas}
 \begin{algorithmic}
 \STATE \textbf{Input:} Current skeleton $s_{0:n}$.
 \STATE \textbf{Output:} New skeleton $s^\prime_{0:n}$.
 \STATE Set $\xi^N_{0:n}=s_{0:n}$.
 \STATE Compute weights $w_k^N=g_k(s_k)$ for $k=0,...,n$.
 \FOR {$i=1,\dots, N-1$}
 \STATE Draw $\xi_0^i\sim \nu(\cdot)$.
 \STATE Compute corresponding weights $w_0^i=g_0(\xi_0^i)$.
 \ENDFOR
 \FOR {$k=1,\dots,n$}
 \STATE \textbf{(Resampling)} 
 \STATE Draw ancestors \[\{a^i_k\}_{i=1}^N\sim Multinomial\left(N,\frac{w_{k-1}^1}{\sum_j w_{k-1}^j},\dots,\frac{w_{k-1}^N}{\sum_j w_{t-1}^j}\right)\]
 and set
 \[\tilde\xi_{0:k-1}^i=\xi_{0:k-1}^{a^i_k}\;.\]
 \STATE \textbf{(Ancestor resampling)}
 \STATE Draw $J$ with probability
 \[\mathbb{P}(J=i)=\frac{w_{k-1}^i P(\tilde \xi_{k-1}^i,s_k)}{\sum_{j=1}^N w_{k-1}^j P(\tilde \xi_{k-1}^j,s_k)}\;,\]
 and set
 \[\xi^N_{0:n}=(\tilde\xi_{0:k-1}^J,s_{k:n})\;.\]
 \STATE \textbf{(Propagation step)}
 \FOR {$i=1,\dots, N-1$}
 \STATE Draw $\xi^i_k\sim P(\tilde\xi_{k-1}^i,\cdot\;)$ and set $\xi_{0:k}^i=(\tilde\xi_{0:k-1}^i,\xi_k^i)$.
 \STATE Compute weights $w_i^k=g_k(\xi_k^i)$.
 \ENDFOR
 \ENDFOR
 \STATE Draw $I$ with probability
 \[\mathbb{P}(I=i)=\frac{w_{n}^i  }{\sum_{j=1}^N w_{n}^j}\;,\]
 and set
 \[s^\prime_{0:n}=\xi^I_{0:n}\;.\]
 \end{algorithmic}
\end{algorithm}
\begin{remark}
 There exists also a Metropolis type of particle MCMC algorithm (PMH). However, straightforward application of PMH within MCMC algorithm for 
 Markov jump processes is not possible, because the acceptance probability in PMH depends on estimates of the normning constant $Z_n$ 
 obtained in two consecutive steps. In the case of 
 Markov jump processes, the dimension  of the state space ($n+1$) can be different in every step, because virtual jumps are
 either added or removed. Therefore to implement PMH in this context, transdimensional Metropolis type moves would be needed.
\end{remark}

\newpage

\section{Numerical experiments}\label{sec:num}  

In this section we present results of simulations which demonstrate the efficiency of the proposed algorithm. 
We concentrate on the case of CTBNs. We start with a toy example with only two nodes joined by a single arrow $(X\to Y)$, i.e.\ the simplest
hidden Markov model with continuous time.
Next we consider a CTBN with a chain structure. The last example is the Lotka-Volterra predator-prey model.
To the best of our knowledge the algorithm of \citet{RaoTeh2013a} is the most efficient method to deal with hidden continuous 
time Markov processes. For this reason we use it in our comparisons.
In the simulations we use RCPP \citep{eddelbuettel2011rcpp} implementation of algorithms. All our codes are available at a request.
\subsection{Toy example}

Consider a CTBN presented in Figure~\ref{fig:toy}. Both nodes $(X,Y)$ have 
two possible states, $\{1,2\}$.  Node $Y$ is fully observed, $X$ is hidden (apart from the beginning and the end of its trajectory).
The transition intensities of $X$ are independent of $Y$ and given by 
\[Q_X=\left(\begin{array}{cc}
             -10&10\\
             10&-10
            \end{array}\right)\;.\]
The conditional intensities of $Y$ are given by
\[Q_{Y|X=1}=\left(\begin{array}{cc}
             -10&10\\
             10&-10
            \end{array}\right),\quad Q_{Y|X=2}=\left(\begin{array}{cc}
             -100&100\\
             100&-100
            \end{array}\right)\;.\]

\begin{figure}
\centering
 \includegraphics[width=.8\textwidth]{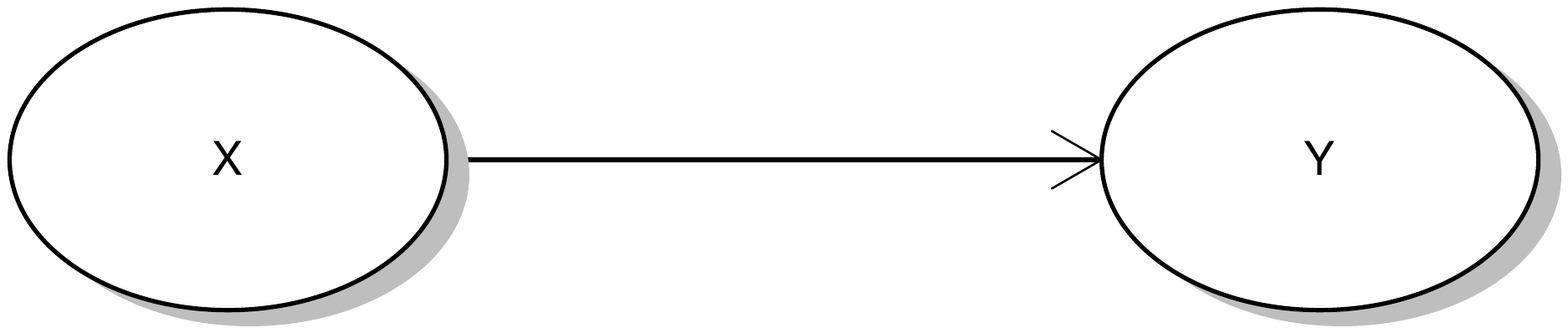}
  \caption{Toy example}
  \label{fig:toy}
\end{figure} 
Assume that we observe the full trajectory of $Y$ in time interval $[0,1]$ and we also observe $X$ at moments $0$ and $1$. 
Our goal is to sample from the posterior distribution of $X([0,1])$ given the evidence $(X(0),X(1),Y([0,1]))$. 
We run our particle MCMC algorithm in two versions. Times of virtual jumps are added via uniformization with $\lambda=20$ in the first version
and distributed according to the homogeneous Poisson process with intensity $\theta=10$ in the second version. 
In both cases the expected number of virtual jumps is approximately the same. We run our algorithm with PGAS with $2$ and $4$ particles.
For a comparison we apply \citet{RaoTeh2013a} algorithm with FFBS subroutine. The actual (hidden) trajectory $X([0,1])$,
the posterior means and standard deviation of MCMC approximation are presented in Figure \ref{fig:toyplot}.
These results are based on $100$ replications, each MCMC run of length $1000$ with a burn-in time $100$.
\begin{figure}
 \centering
 \includegraphics[width=.8\textwidth]{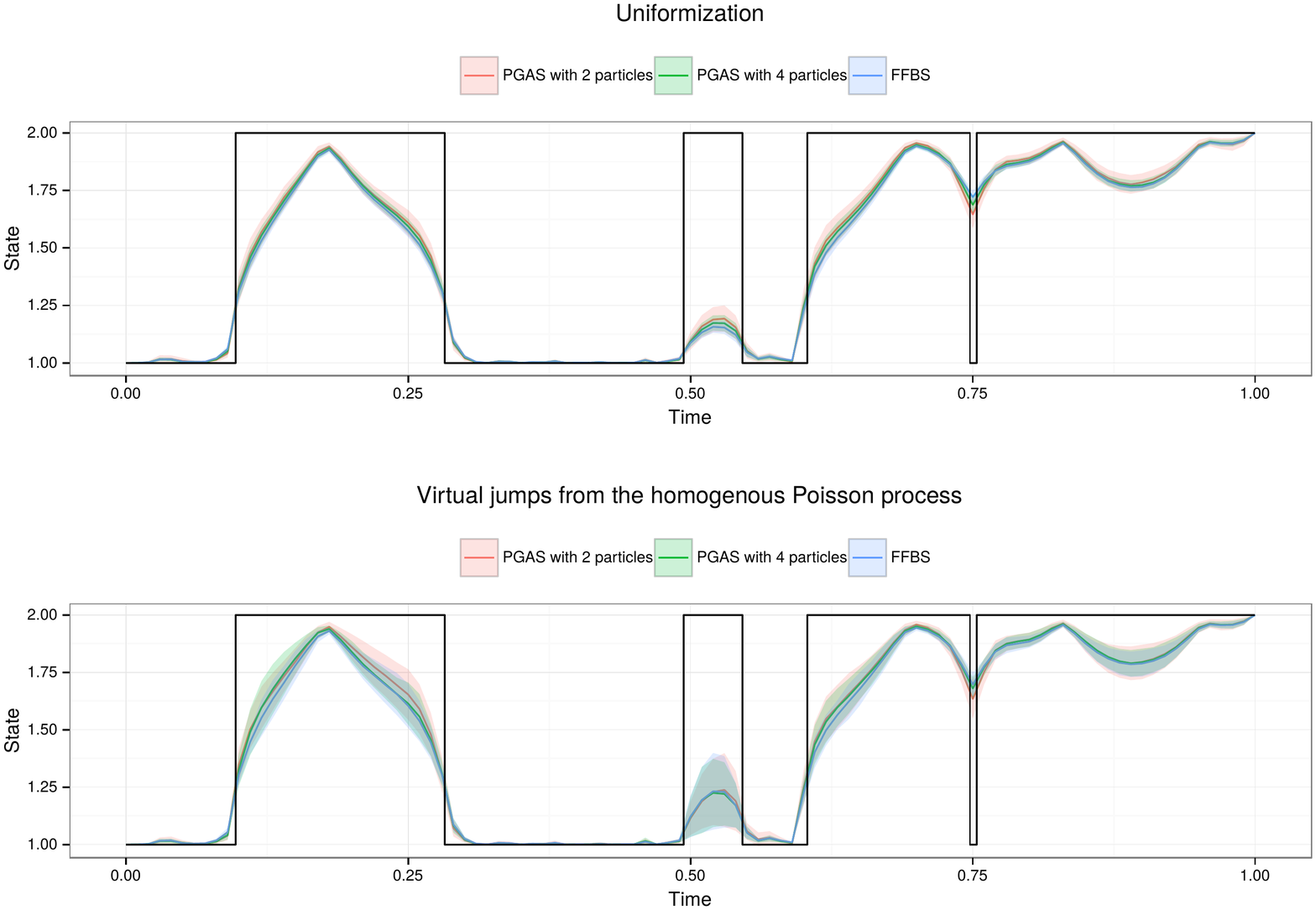}
  \caption{Posterior mean and standard deviation of MCMC approximation for the toy example. }
  \label{fig:toyplot}
\end{figure}
We observe that the estimated trajectory is very similar for all the compared methods. 
The difference is in the variance and consequently in the width of the ``confidence bands''.
In this example the algorithm with uniformization outperforms the algorithm with homogeneous virtual jumps.
As it was expected, our algorithms with PGAS have the variance greater than those with FFBS but 
the difference between PGAS with $4$ particles and FFBS is not too big. 
\begin{figure}
 \centering
 \includegraphics[width=.8\textwidth]{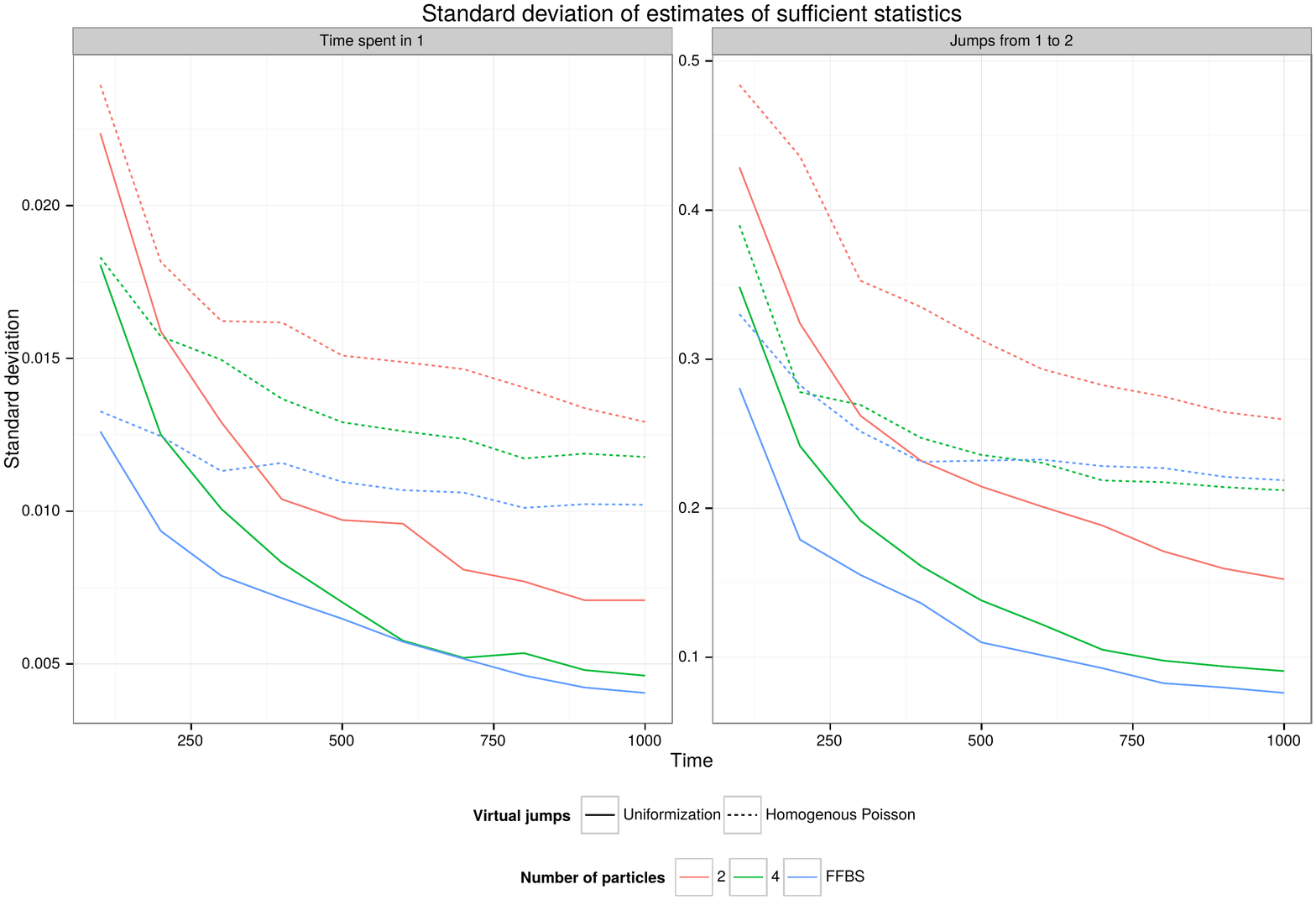}
  \caption{Standard deviation of sufficient statistics versus number of iterations of MCMC algorithm.}
  \label{fig:toystat}
\end{figure}
To analyze the rate of convergence of the algorithms we also compute standard deviation of sufficient statistics 
(number of jumps and time spent in each state). The results are shown in Figure~\ref{fig:toystat}. 
Again we obtain similar conclusions. It is clearly visible that the algorithms with FFBS have lower variance than those with PGAS. 
However, the difference is rapidly decreasing with increasing number of particles.
Approximately the cost of FFBS is of order $\mathcal{O}(|\X|^2 \mathbb{E}(n))$ and the cost of PGAS is $\mathcal{O}(N \mathbb{E}(n))$. 
Hence in the example under consideration, we obtain comparable quality of estimation for the algorithms with FFBS and PGAS which have 
the same computational cost.

\subsection{Chain network}

The model considered in this subsection is based on \citep[Subsection 5.5]{RaoTeh2013a}.
It is a network which consists of $M$ nodes equipped with the ``chain'' graph $1\to 2\to\cdots\to M$.  For each node the set of possible states is
$\{1,\dots, S\}$. 
The transition intensities of every node, except the first one, depend on current state of the previous node.
Namely, off-diagonal elements of intensity matrix are given by
                 
\begin{align*} Q_1(x_1,{x_1^\prime})&=\begin{cases}
                    \frac{1}{2}&\text{if } x_1^\prime= (x_1+1)\ \rm {mod}\ S\;;\\
                    \frac{1}{2(S-2)}&\text{ otherwise}\;,
                   \end{cases} \\
                   Q_m(x_{m-1};x_m,{x_m^\prime})&=\begin{cases}
                    \frac{1}{S-1} &\text{ if } x_m=x_{m-1}\;, \\
                    1             &\text{ if }  x_{m-1}\not=x_m \text{ and } x^\prime_m=x_{m-1} \;;\\
                    \frac{1}{S-2}&\text{ otherwise}.  
                   \end{cases}
                   \end{align*} 
In words, the head node leaves the current state $x_1$ with intensity $1$ and prefers to jump to $x_1+1\text{ mod } S$. For any other node, 
if its current state $x_m$ differs from that of parent state $x_{m-1}$ then intensity of jumping is $2$ and the process 
prefers to jump to $x_{m-1}$. If $x_m=x_{m-1}$ then the process leaves the current state with intensity $1$ and 
chooses a new state at random.
We observe the process at the beginning and at the end of time interval $[0,T]$. 

We run two MCMC algorithms targeting the posterior distribution over 
latent trajectories. Our algorithm with PGAS based on $N=10$ particles is compared with Rao and Teh's algorithm with FFBS.  
Both the algorithms use  uniformization, with $\lambda$ equal to twice the intensity of leaving the current state. 

We begin with a network with $M=3$ nodes, $S=2$ states of every node and time length $T=5$. In the first series of our experiments,
we increase the number of nodes to $M=6,12,24$. In the second series we increase the length of time interval to $T=10,20,40$. Finally
we change the size of the state space to  $S=10,50,100$. For each of the scenarios we run 20 replications of each of
the MCMC algorithms. 
We use CODA R package \citep{coda} to estimate the effective sample size of the following  statistics: time spent in each state and number of jumps,
for nodes $m=1,\dots,M$.  The median of all these quantities serves as an estimator of ESS for the whole CTBN. 
\begin{figure}
 \centering
 \includegraphics[width=.8\textwidth]{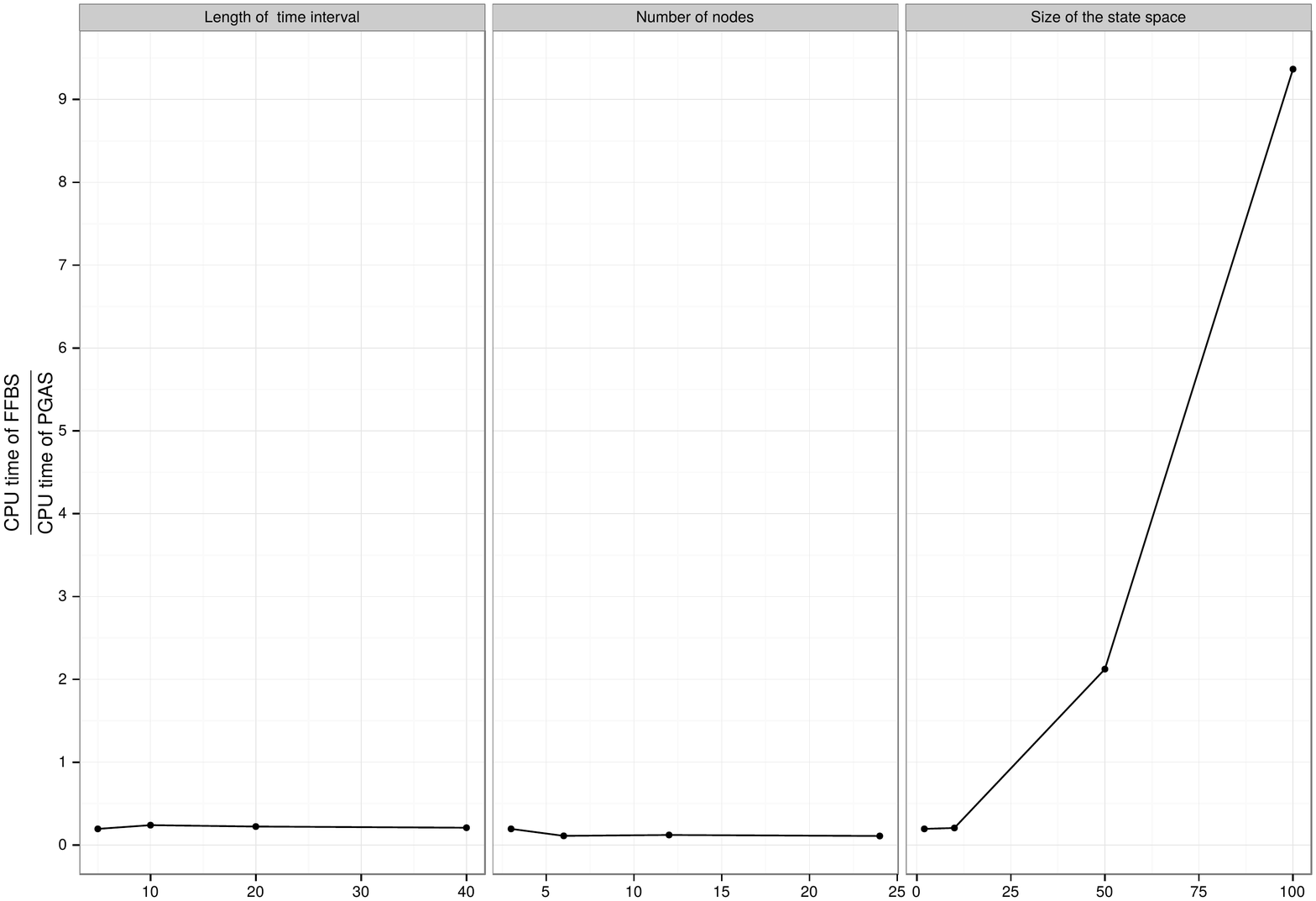}
  \caption{Ratio of CPU time (FFBS/PGAS) needed to generate a sample with $ESS=100$ for the chain model: with increasing length of time interval (left),
  increasing number of nodes (center), increasing size of the state space (right). }
  \label{fig:ratio}
\end{figure}
In Figure~\ref{fig:ratio} we present the ratio of CPU time needed to generate a sample with ESS equal to $100$. 
In the beginning, when state space is of size $S=2$, FFBS is more effective,  around 4-5 times as fast as PGAS. 
This is what should be expected, since the cost of PGAS with $10$ particles is higher than the cost of FFBS for a small 
space size. Also as expected, the ratio does not change significantly if the number of nodes increases. 
The same seems to happen if we increase the length of time interval. 
This last fact is slightly surprising, because the rate of convergence of particle Gibbs depends on the number of jumps, see 
\citet{andrieu2013uniform,lindsten2014uniform}. However, if the size of the state space increases then the cost of our 
proposed method becomess significantly lower than that
of \citet{RaoTeh2013a} approach. For instance if $S=50$ then our algorithm is twice as fast as  \citet{RaoTeh2013a} and for $S=100$ 
it is more than 9 times as fast. Experiments with a different number of particles ($N=5,20,50$) lead to the same conclusions. 

\subsection{Lotka-Volterra model}
\begin{figure}
 \centering
 \includegraphics[width=.8\textwidth]{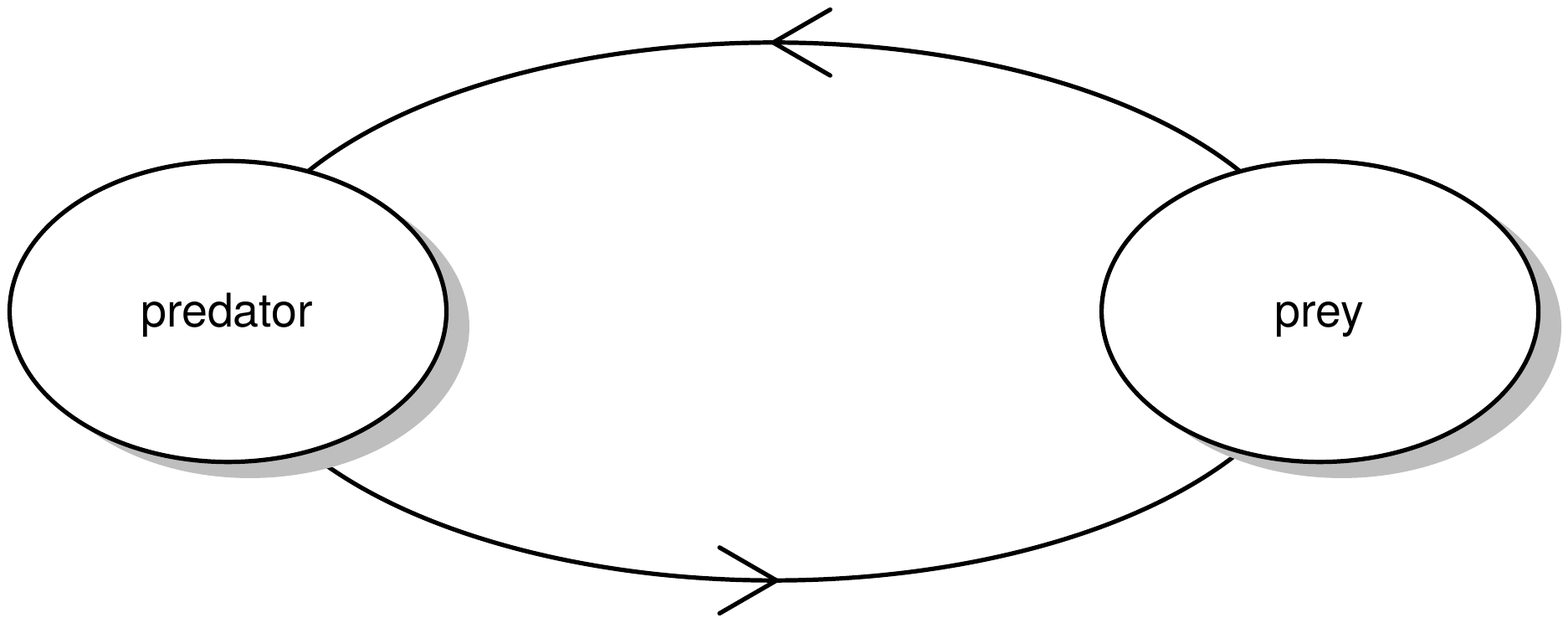}
  \caption{Predator-Prey model as CTBN. }
  \label{fig:lv}
\end{figure}

The last example is Lotka-Volterra model \citep{wilkinson2009stochastic,opper2008variational}, which describes evolution of two interacting populations 
of prey and predator species. This process can be viewed as the two node CTBN shown in Figure~\ref{fig:lv}. 
Let $x$ and $y$ represent the size of prey and predator populations, respectively.
The transition intensities are given by
\begin{align*}
 Q(\{x,y\},\{x+1,y\})&=\alpha x\;,\quad Q(\{x,y\},\{x-1,y\})=\beta xy\;,\\
 Q(\{x,y\},\{x,y+1\})&=\delta xy\;,\quad Q(\{x,y\},\{x,y-1\})=\gamma y\;,
\end{align*}
All other intensities are $0$.  The state space is infinite: $\{0,1,\ldots\}\times \{0,1,\ldots\}$.
Following \citet{RaoTeh2013a} we set the parameters as follows:
$\alpha=\gamma=5\times10^{-4}$ and $\beta=\delta=1\times10^{-4}$. 
We condsider the process in time interval $[0,3000]$ with known initial position and noisy observations $Y(t)$ at discrete times uniformly spaced in interval
$[0,1500]$ with the likelihood given by
\[p(Y(t)|X(t))\propto \left[ 2^{|X(t)-Y(t)|}+10^{-6}\right]^{-1}\;.\]
Given this evidence, we estimate the posterior distribution over sample paths of $X$ using our sampler.
We compute estimates of $X$ in the interval $[0,1500]$ (where observations $Y$ are available) and also predict values of $X$ in the interval 
$[1500,3000]$.  Due to unboudedness of intensities we are not able to 
use uniformization technique and so we use homogenous virtual jumps with $\theta=30$. For this choice of $\theta$, the number of virtual jumps is 
approximately equal to the number of true jumps. We run MCMC simulation of length $1000$ with $100$ initial iterations 
treated as burn-in time. 
In the PGAS we use $100$ particles.  The results of simulation are given in Figure~\ref{fig:lvplot}. 
\begin{figure}
 \centering
 \includegraphics[width=.8\textwidth]{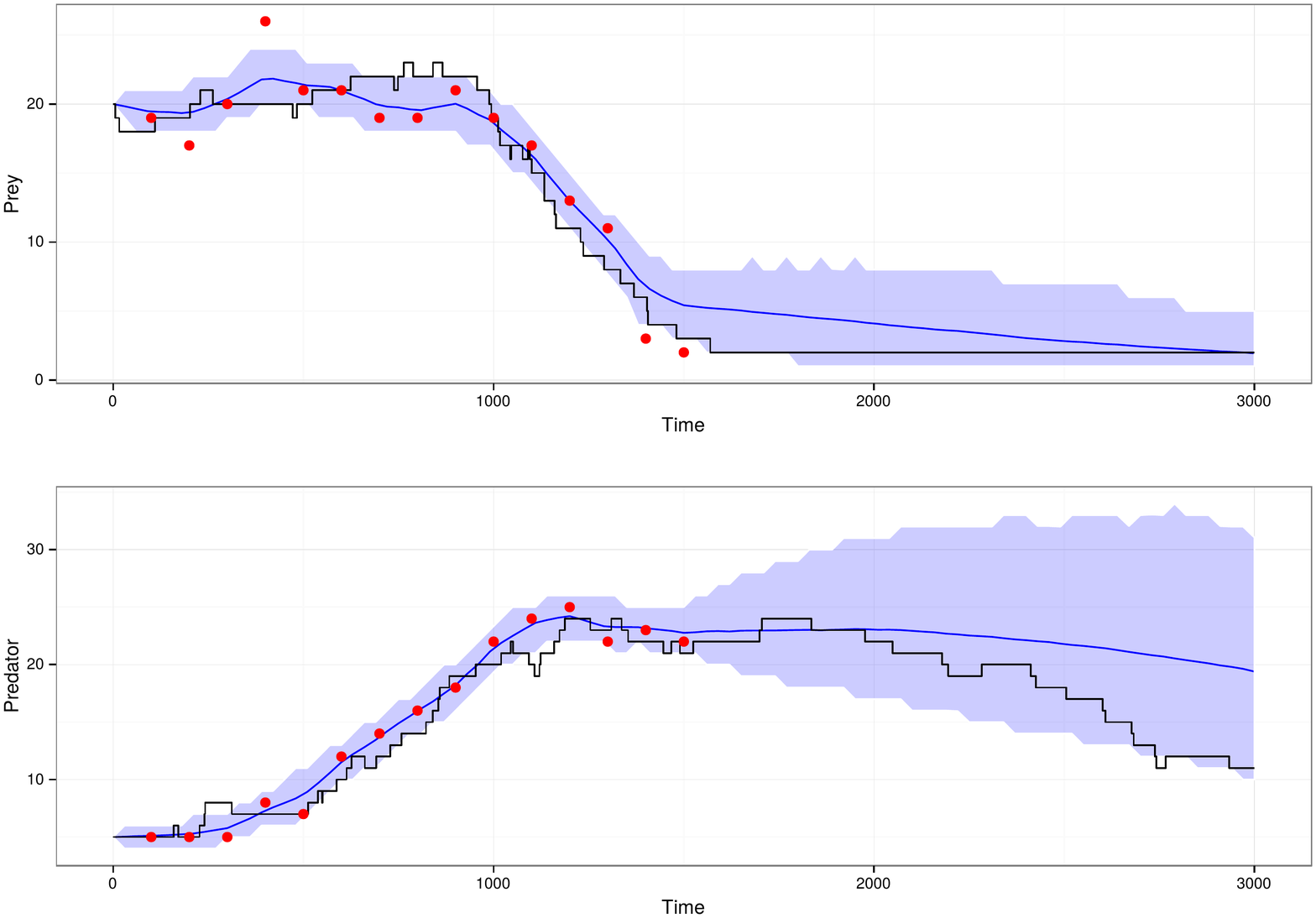}
  \caption{Resuts of MCMC aproximation for Predator-Prey model. Black line - true path, blue line - posterior mean, shadow - 90\% credible interval,  red points - observations. }
  \label{fig:lvplot}
\end{figure}
We conclude that the quality of estimates is the same as in \citet{RaoTeh2013a}.  
Note that in opposition to Rao and Teh's algorithm we need not truncate the state space and the computational cost of our algorithm is significantly lower, 
namely $\mathcal{O}(100\mathbb{E}(n))$ for our method and $\mathcal{O}(200^2\mathbb{E}(n))$ for Rao and Teh, with state space truncted to 
$\{0,\dots,200\}\times\{0,\dots,200\}$ as suggested by these authors.

\section{Conclusions}\label{sec:con}
In the present paper we propose a new MCMC algorithm for sampling from the posterior distribution of hidden trajectory of
a Markov jump process. 
The general idea is the same as in \citet{RaoTeh2013a}, namely we alternately add virtual jumps and update the skeleton
of the process. The main novelty is that our algorithm uses PGAS to sample the skeleton, while Rao and Teh use FFBS.
Thus instead of sampling exactly from a conditional distribution, we make a step of a Markov chain which preserves this 
distribution. This modification has some disadvantages, as slower mixing of the entire procedure. However, there are also
important advantages of our approach. Unlike previous methods our algorithm can be implemented 
even if the state space is infinite.  The cost of a single step of the  proposed algorithm 
does not depend on the size of the state space. Consequently, we  can recommend our algorithm for problems where the space
is either infinite or finite but large. If the size of the state space is not big, then our algorithm with PGAS converges slightly 
slower than the algorithm with FFBS. However, the difference between them is rapidly diminishing if we increase the number     
of particles in PGAS. 

In the present paper we describe an algorithm for time homogeneous processes, only to avoid too many technical details.  However,
the generalization to non-homogeneous processes is rather straightforward. 
For details of adding virtual jumps in non-homogeneous case we refer to
\citet{rao2012mcmc}. Since PGAS can easly deal with general state spaces, our algorithm can also be applied to piecewise deterministic 
Markov jump processes on general state spaces (i.e.\ processes which 
evolve deterministically between jumps and move according to some Markov kernel at moments of jumps). 

We note that an important issue is to find an optimal number of virtual jumps. Small number of virtual jumps can lead to poor mixing 
of moments of jumps. In the case of PGAS, large number of virtual jumps not only increases computational cost, as it is in the case of FFBS, 
but may have a negative impact on mixing of the whole algorithm. It is because the convergence rate of particle Gibbs depends on 
the length of simulated trajectory.

\appendix
\section*{Appendix A.}
\label{app:proofs}
\begin{proof}[Proof of Proposition~\ref{prop:density}] We are to check that if $(\tilde{T},\tilde{S})$ has the distribution given by 
\eqref{eq:density_with_virtual} then $(\tilde{T}_J,\tilde{S}_J)$ is distributed according to \eqref{eq:density}.
First we compute the distribution of waiting time for the next true jump. Without loss of generality we can assume that the 
previous jump occurred at time $0$ and 
$X(0)=s$. To get the first true jump we generate subsequent moments of potential jumps 
$\tilde t_1,\dots,\tilde t_i\dots$  such that $\tilde t_i-\tilde t_{i-1}$ are i.i.d.\ from $Exp(R(s))$. 
The candidate is accepted with probability ${Q(s)}/{R(s)}$. Hence
\begin{align*}
 \mathbb{P}(t_1\leq t)&=\sum_{k=1}^\infty\mathbb{P}(\tilde t_k<t)\left(1 -\frac{Q(s)}{R(s)}\right)^{k-1}\frac{Q(s)}{R(s)}\\
                      &=\sum_{k=1}^\infty\mathbb{P}\left(\sum_{i=1}^k(\tilde t_i-\tilde t_{i-1})<t\right)\left(1 -\frac{Q(s)}{R(s)}\right)^{k-1}\frac{Q(s)}{R(s)}\\
&=\sum_{k=1}^\infty\int_0^t\frac{R(s)^k}{(k-1)!}u^{k-1}\exp\{-R(s)u\}\rm d u\left(1 -\frac{Q(s)}{R(s)}\right)^{k-1}\frac{Q(s)}{R(s)}\\
&=\int_0^t\sum_{k=1}^\infty\frac{(R(s)-Q(s))^{k-1}u^{k-1}}{(k-1)!}\exp\{-R(s)u\}\d u\; Q(s)\\
&=\int_0^t\exp\left\{\left[R(s)-Q(s)\right]u-R(s)u \right\}\d u\; Q(s)\\
&=\int_0^t Q(s)\exp\{-Q(s)u\}\d u =1- \exp\{-Q(s)t\}\;.
\end{align*}
We have obtained an expression which is exactly the c.d.f.\  of waiting time for the next jump of process with intensity matrix $Q$.
To conclude the proof, it is enough to note that {
\begin{align*}
\mathbb{P}(s_1=s^\prime|s_0= s)&=\mathbb{P}(\tilde s_1=\tilde s^\prime|\tilde s_0=\tilde s,\tilde s_0\neq\tilde  s_1)\\
&=\frac{{Q(\tilde s,\tilde s^\prime)}/{R(\tilde s)}}{{Q(\tilde s)}/{R(\tilde s)}}=\frac{Q(\tilde s,\tilde s^\prime)}{Q(\tilde s)}\;.
\end{align*} }
\end{proof}

\begin{proof}[Proof of Corollary~\ref{cor:density}]
By construction of the thinning procedure and by Proposition~\ref{prop:density} we have 
\begin{align*}
p(V_j|X(t_{j-1})=s,t_{j-1},t_j)&=\frac{p(V_j,t_{j-1},t_j|s)}{p(t_{j-1},t_j|s)}\\&=
\frac{R(s)^{|V_j|+1}\displaystyle\frac{(R(s)-Q(s))^{|V_j|}Q(s)}{R(s)^{|V_j|+1}}\exp\{-(t_j-t_{j-1})R(s)\}}{Q(s)\exp\{-(t_j-t_{j-1})Q(s)}\\
&=(R(s)-Q(s))^{|V_j|}\exp\left\{-(t_{j}-t_{j-1}) (R(s)-Q(s))\right\} \;.\phantom{\displaystyle A^{A^{A^{A}}}}
\end{align*}
\end{proof}

\bibliography{refs}

\end{document}